\documentclass{llncs}

\usepackage{graphicx}
\usepackage{amsmath}
\usepackage{amssymb}
\usepackage{epic,eepic}
\usepackage{url}
\usepackage{yhmath}
\usepackage{multirow}

\newcommand{\define}{\,\widehat{=}\, }
\newcommand{\xx}{\mathbf{x}}
\newcommand{\uu}{\mathbf{u}}
\newcommand{\fb}{\mathbf{f}}
\newcommand{\<}{\langle}
\newcommand{\rg}{\rangle}
\newcommand{\ud}{\mathrm{d}}
\newcommand{\tr}[1]{\mathcal{T}\!r(\mathcal{#1})}
\newcommand{\oomit}[1]{}

\begin{document}

\pagestyle{plain}

\title{Synthesizing Switching Controllers for Hybrid Systems by Continuous Invariant Generation}
\author 
{Deepak Kapur\inst{1} \and Naijun Zhan\inst{2} \and Hengjun Zhao\inst{2,3}}
\institute
{ Dept.~of Comput. Sci., University of New Mexico,
  Albuquerque, NM, USA\\
  \email{kapur@cs.unm.edu}
  \and
  State Key Lab. of Comput. Sci.,
  Institute of Software, CAS, Beijing, China\\
  \email{znj@ios.ac.cn}
  \and
  University of Chinese Academy of Sciences, Beijing, China\\
  \email{zhaohj@ios.ac.cn}
}
\maketitle

\begin{abstract}
We extend a template-based approach for synthesizing switching controllers
for semi-algebraic hybrid systems, in which all expressions are
polynomials. This is achieved by combining a QE (quantifier elimination)-based method for
generating continuous invariants with a qualitative approach for predefining templates. Our synthesis method is relatively complete with regard to a given family of predefined templates. Using qualitative analysis, we discuss heuristics to reduce the numbers of parameters appearing in the templates. To avoid too much human interaction in choosing templates as well as the high computational complexity caused by QE, we further investigate applications of the SOS (sum-of-squares) relaxation approach and the template polyhedra approach in continuous invariant generation, which are both well supported by efficient numerical solvers.
\end{abstract}

\section{Introduction}

Hybrid systems, in which computations proceed both by continuous evolutions as well as discrete jumps simulating transition from one mode to another mode, are often used to model devices controlled by computers in many application domains \cite{Alur11}. Combining ideas from state machines in computer science and control theory, formal analysis, verification and synthesis of hybrid systems have been an important area of active research. In verification problems, a given hybrid system is required to satisfy a desired safety property e.g. that the temperature of a nuclear reactor will never go beyond a maximum threshold, as it may cause serious economic, human and/or environmental damage, thus implying that the system will never enter any unsafe state. A synthesis problem is harder given that the focus is on designing a controller that ensures the given system will satisfy
a safety requirement, reach a given set of states, or meet an optimality criterion, or a desired combination of these requirements.

Automata-theoretic and logical approaches have been primarily
used for verification and synthesis of hybrid systems
\cite{ACHH93,Pnueli00,Sastry00}. In \cite{Pnueli00,Sastry00}, a
general framework for controller synthesis based on hybrid
automata was proposed, which relies on \emph{backward reachable set} computation and \emph{fixed point iteration}. Two restrictions of this approach are (i) the computation of backward reachable set is hard for most continuous dynamics, and (ii) termination of the fixpoint iteration procedure cannot be guaranteed, even for those hybrid systems whose backward reachable sets are easily computable. Thus most of the research, e.g. \cite{Jha09}, focuses on overcoming the above two restrictions.

Recently, a deductive approach for verification and synthesis
based on constraint solving was proposed in \cite{Tiwari08,Platzer08,PlatzerClarke08,Taly11,aplas,Sturm-Tiwari-issac11}. The central idea is to reduce verification and synthesis problems of hybrid systems to invariant generation problems, much like verification of programs.
As proposed in \cite{Kapur03,Kapur06,SankaranarayananSipmaMannaPOPL04}, if invariants are hypothesized to be of certain shapes, then corresponding templates with associated parameters can be used and the invariant generation problem can be reduced to constraint solving over parameters by quantifier elimination. This methodology is used in \cite{Taly11}
for synthesizing switching controllers meeting safety requirements, while in \cite{Taly-10-reach}, the approach is extended for satisfying both safety and reachability requirements. A common problem with template-based method is that it heavily relies on a user specifying the shape of invariants that are of interest, thus making it interactive and user driven, raising doubts about its scalability and automation. Besides, the inference rules for inductive invariants in \cite{Taly09,Taly11,Taly-10-reach} are sound and complete for classes of invariants, e.g. smooth, quadratic and convex invariants, but are not complete for generic semi-algebraic sets.

Inspired by \cite{Kapur97,Pnueli00,Taly11} and \cite{emsoft11}, we extend in this paper the template-based invariant generation approach for synthesizing switching controllers of hybrid systems to meet given safety requirements. The paper makes the following contributions:
\begin{itemize}
  \item We formalize the solution to switching controller synthesis problem (Problem \ref{dfn:prblm}) of hybrid systems in terms of continuous invariants (see Theorem \ref{thm:main}), and thus lay the foundation of the synthesis method based on continuous invariant generation using templates and constraint solving.
  \item In the QE-based synthesis framework we use the method for continuous invariant generation proposed in \cite{emsoft11} as an integral component. This method is proved in \cite{emsoft11} to be sound and relatively complete with respect to a given shape of invariants (i.e. a given family of predefined templates). As a result, in contrast to the methods used in \cite{Taly09,Taly11,Taly-10-reach}, there is more flexibility in our approach because of the possibility of discovering all possible invariants of the given shape.
    \item Using the qualitative approach proposed in \cite{Kapur97} for analyzing continuous evolution in certain modes of a hybrid system, we can develop heuristics to determine a more precise shape of templates to be used as invariants, thus reducing the numbers of parameters appearing in templates.
    \item We further improve the degree of automation and scalability of the template-based method in two ways: (i) for general polynomial templates, using \emph{sum-of-squares} (SOS) relaxation, the constraint on parameters appearing in templates is transformed into a \emph{semi-definite program}(SDP), which is convex and thus can be solved efficiently; (ii) for linear systems and a special type of templates---template polyhedra, again by sacrificing relative completeness, the continuous invariant generation problem can be reduced to a BMI (\emph{bilinear matrix inequality}) feasibility problem, which is also much easier to solve (numerically) than QE.
  \end{itemize}

\subsection*{Related Work}
Our work in this paper resembles \cite{Taly11} but differs in that: i) our method is cast in the setting of hybrid automata and searches for a family of continuous invariants that refine the original domains, rather than a single global controlled invariant; ii) a sound and complete criterion is used in continuous invariant generation; iii) various techniques are applied for scalability.

The SOS relaxation approach has been successfully used in safety verification of hybrid systems.
In \cite{Prajna04,Prajna07}, the authors used the SOSTOOLS software package \cite{SOSTOOLS} to compute \emph{barrier certificates} for polynomial hybrid systems. In \cite{yzf11,yzf12}, the authors proposed a hybrid symbolic-numeric approach to compute exact inequality invariants of hybrid systems, by first solving (bilinear) SOS programming numerically and then applying \emph{rational vector recovery} techniques.

A necessary and sufficient condition for positive invariance of convex polyhedra for linear continuous systems was provided in \cite{Castelan93}. This condition is extended to linear systems with open polyhedral domain for our need in the paper. Template polyhedra was used in \cite{Sankaranarayanan08tacas,Sankaranarayanan08hscc} to compute positive invariants of hybrid systems by
\emph{policy iteration}, which differs from our treatment of the problem using BMI. Recently, a method for computing polytopic invariants for polynomial dynamical systems using template polyhedra and linear programming was proposed \cite{BenSassi12}.

Mathematical programming techniques and relevant numerical solvers have also been widely applied to static program analysis. Actually, the template polyhedra abstract domain was first proposed in \cite{Sankaranarayanan05} to generate linear program invariants using linear programming. In \cite{Cousot05}, to verify invariance and termination of semi-algebraic programs, verification conditions are abstracted into numerical constraints using Lagrangian relaxation or SOS relaxation, which are then resolved by efficient SDP solvers.

In our recent work \cite{fm2012}, we studied an optimal switching controller synthesis problem arising from an industrial oil pump system with piece-wise constant continuous dynamics. A hybrid approach combining symbolic computation with numerical computation was developed to synthesize safe controllers with better optimal values.

\subsubsection{Paper Structure.}
To be completed.
\newline
1\newline
2\newline
3\newline
4\newline
5\newline
6\newline
7\newline

\section{Problem Description}
Following \cite{Pnueli00,Sastry00},
 we use hybrid automata to model hybrid systems.

\begin{definition}[Hybrid Automaton]\label{dfn:HA}
A hybrid automaton (HA) is a system $\mathcal{H}\,\define\,(Q,X,f,D,E,G)$, where
\begin{itemize}
  \item[$\bullet$]$Q=\{q_1,\ldots,q_m\}$ is a set of discrete states;
  \item[$\bullet$]$X=\{x_1,\ldots,x_n\}$ is a set of continuous state variables, with $\xx=(x_1,\ldots,x_n)$ ranging over $\mathbb R^n$;
  \item[$\bullet$]$f: Q\rightarrow (\mathbb R^n\rightarrow \mathbb R^n)$ assigns to to each discrete state $q\in Q$ a vector field $\fb_q$;
  \item[$\bullet$]$D: Q\rightarrow 2^{\mathbb R^n}$ assigns to each discrete state $q\in Q$ a domain  $D_q\subseteq \mathbb R^n$;
  \item[$\bullet$]$E\subseteq Q\times Q $ is  a set of discrete transitions;
  \item[$\bullet$]$G: E \rightarrow 2^{\mathbb R^n}$ assigns to each transition  $e\in E$ a switching guard  $G_e$ $\subseteq \mathbb R^n$\,.
\end{itemize}
\end{definition}

\begin{remark} For ease of presentation, we make the following assumptions:
  \begin{itemize}
  \item for all $q\in Q$, $\fb_q$ is \emph{polynomial} vector function; besides, $\fb_q$ defines a \emph{complete} vector field, that is, for any $\xx_0\in\mathbb R^n$, the solution to the differential $\dot \xx=\fb_q$ uniquely exists on $[0,\infty)$;
  \item for all $q\in Q$ and all $e\in E$, $D_q$ and $G_e$ are \emph{closed semi-algebraic} sets\setcounter{footnote}{0}\footnote{A set $A\subseteq \mathbb R^n$ is called semi-algebraic if there is a quantifier-free polynomial formula $\varphi$ s.t.
$A=\{\xx\in\mathbb R^n\mid\varphi(\xx)\,\,\mbox{is true}\}$\,.};
  \item the initial condition in each discrete mode is assumed to be identical with the domain, and all reset functions are assumed to be identity mappings.
\end{itemize}
\end{remark}

We use a nuclear reactor system discussed in \cite{ACHHH95,Ho95,Kapur97} as a running example through this paper.
\begin{figure}
\setlength{\unitlength}{1cm}
\begin{center}
\begin{picture}(8,4)(-3,-1.8)
\put(-2,1.8){\oval(2.5,1.3)}
\put(3,1.8){\oval(2.5,1.3)}
\put(-2,-0.9){\oval(2.5,1.3)}
\put(3,-0.9){\oval(2.5,1.3)}

\put(-0.77,1.8){\vector(1,0){2.5}}
\put(1.73,-0.9){\vector(-1,0){2.5}}
\put(3.0,1.14){\vector(0,-1){1.4}}
\put(-2.0,-0.26){\vector(0,1){1.4}}

\put(0.26,1.95){\makebox(0,0)[l]{$\scriptstyle{G_{12}}$}}
\put(0.26,1.66){\makebox(0,0)[l]{$\scriptstyle{p\,=\,0}$}}
\put(0.26,-1.1){\makebox(0,0)[l]{$\scriptstyle{G_{34}}$}}
\put(0.26,-.76){\makebox(0,0)[l]{$\scriptstyle{p\,=\,1}$}}
\put(-2.5,0.5){\makebox(0,0)[l]{$\scriptstyle{G_{41}}$}}
\put(-1.9,0.5){\makebox(0,0)[l]{$\scriptstyle{p\,=\,0}$}}
\put(3.1,0.5){\makebox(0,0)[l]{$\scriptstyle{G_{23}}$}}
\put(2.4,0.5){\makebox(0,0)[l]{$\scriptstyle{p\,=\,1}$}}

\put(-2.6,2.7){\makebox(0,0)[l]{$q_1$\small {:\,no rod}}}
\put(1.8,2.7){\makebox(0,0)[l]{$q_2$\small {:\,being immersed}}}
\put(-2.9,-1.8){\makebox(0,0)[l]{$q_4$\small {:\,being removed}}}
\put(2.3,-1.8){\makebox(0,0)[l]{$q_3$\small {:\,immersed}}}

\put(-3,2.1){\makebox(0,0)[l]{$\scriptstyle{\dot x\,=\,x/10-6p-50}$}}
\put(-3,1.8){\makebox(0,0)[l]{$\scriptstyle{\dot p\,=\,0}$}}
\put(-3,1.5){\makebox(0,0)[l]{$\scriptstyle{D_1\,\define\,p\,=\,0}$}}

\put(2,2.1){\makebox(0,0)[l]{$\scriptstyle{\dot x\,=\,x/10-6p-50}$}}
\put(2,1.8){\makebox(0,0)[l]{$\scriptstyle{\dot p\,=\,1}$}}
\put(2,1.5){\makebox(0,0)[l]{$\scriptstyle{D_2\,\define\,0\,\leq p\,\leq 1}$}}

\put(2,-.6){\makebox(0,0)[l]{$\scriptstyle{\dot x\,=\,x/10-6p-50}$}}
\put(2,-.9){\makebox(0,0)[l]{$\scriptstyle{\dot p\,=\,0}$}}
\put(2,-1.2){\makebox(0,0)[l]{$\scriptstyle{D_3\,\define\,p\,=\,1}$}}

\put(-3,-0.6){\makebox(0,0)[l]{$\scriptstyle{\dot x\,=\,x/10-6p-50}$}}
\put(-3,-0.9){\makebox(0,0)[l]{$\scriptstyle{\dot p\,=\,-1}$}}
\put(-3,-1.2){\makebox(0,0)[l]{$\scriptstyle{D_4\,\define\,0\,\leq p\,\leq 1}$}}
\end{picture}
\caption{Nuclear reactor temperature control.}\label{fig:nuclear}
\end{center}
\end{figure}
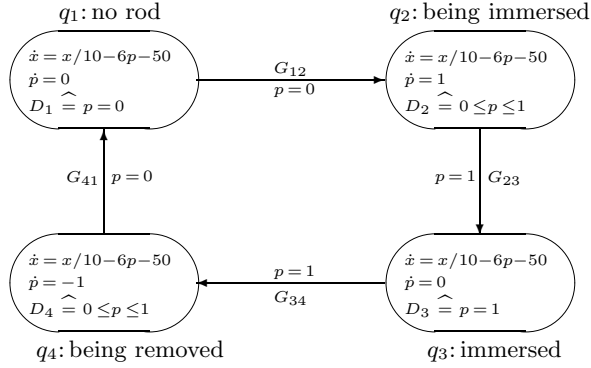
\begin{example}\label{ex:nuclear0}
The nuclear reactor system consists of a reactor core and a cooling rod which is immersed into and removed out of the core periodically to keep the temperature of the core, denoted by $x$, in a certain range.
Denote the fraction of the rod immersed into the reactor by $p$. Then the
initial specification of this system can be represented using the hybrid automaton in Fig. \ref{fig:nuclear}.
\end{example}

The semantics of a hybrid automaton $\mathcal{H}$ can be defined by the set of trajectories it accepts. For the formal definitions of \emph{hybrid time set} and \emph{hybrid trajectory} the readers are referred to \cite{Sastry00}.
We denote the set of trajectories of $\mathcal{H}$ by $\tr{H}$, ranged over $\omega,\omega_1,\ldots$.
All trajectories of $\mathcal{H}$ starting from the initial state $(q_0,\xx_0)$ is denoted by $\tr{H}(q_0,\xx_0)$.

The domain of a hybrid automaton $\mathcal H$ is defined as  $D_{\mathcal{H}}\,\define\,\bigcup_{q\in Q}(\{q\}\times D_q)$. We call $\mathcal H$ \emph{non-blocking} if for any $(q,\xx)\in D_{\mathcal H}$, there is a hybrid trajectory from $(q,\xx)$ which can either be extended to infinite time $t=\infty$ or execute infinitely many discrete transitions; otherwise $\mathcal H$ is called \emph{blocking}.

A \emph{safety requirement} $S$ assigns to each mode $q\in Q$
a safe region $S_q\subseteq \mathbb R^n$, i.e. $S=\bigcup_{q\in
Q}(\{q\}\times  S_q)$. Alternatively, there could be a global
safety requirement $S$ which all modes are required to satisfy.

According to \cite{Pnueli00}, the switching controller synthesis problem with regard to a given safety requirement can be formally defined as follows:

\begin{problem}[Controller Synthesis for Safety]\label{dfn:prblm}
Given a  hybrid automaton  $\mathcal{H}$ and a safety property $S$, find a hybrid automaton $\mathcal H'=(Q,X,f,D',E,G')$ such that
\begin{itemize}
  \item[(r1)]Refinement: for any $q\in Q$, $D'_q\subseteq D_q$, and for any $e\in E$, $G'_e\subseteq G_e$;
  \item[(r2)]Safety: for any trajectory $\omega$ that $\mathcal H'$ accepts, if $(q,\xx)$ is on $\omega$, then $\xx\in S_q$;
  \item[(r3)]Non-blocking: $\mathcal H'$ is non-blocking.
\end{itemize}
\end{problem}

If such $\mathcal H'$ exists, then $\mathcal{SC}\,\define\,\{G'_e\subseteq   \mathbb R^n\mid e\in E\}$ is a \emph{switching controller} satisfying safety requirement $S$, and $D_{\mathcal H'}\,\define\,\bigcup_{q\in Q}(\{q\}\times D'_q)$ is the \emph{controlled invariant set} rendered by $\mathcal {SC}$.

\section{A QE-Based Approach} \label{sec:synthesis}
\subsection{Continuous Invariant}
  Along the line of \cite{Taly11}, we consider the switching controller synthesis problem by combining a relative complete method for generating continuous invariants in \cite{emsoft11}, and heuristics for predefining templates for these invariants using qualitative analysis in \cite{Kapur97}.
Below, we review the concept of \emph{continuous invariant} used
in \cite{emsoft11} based on a related concept in \cite{PlatzerClarke08}.

\begin{definition}[Continuous Invariant (CI)]\label{dfn:DI}
Given a mode $q \in Q$ in a hybrid automaton $\mathcal{H}$, a set
$P\subseteq \mathbb R^n$ is called a continuous invariant of
$(D_q,\fb_q)$ if for all $\xx_0\in P\cap D_q$ and all $T\geq 0$, the solution $\xx(t)$\footnote{We assume the existence and uniqueness of solutions is guaranteed by $\fb_q$.} of $\dot \xx=\fb_q(\xx)$ over $[0,T]$ with $\xx(0)=\xx_0$ satisfies
$$(\forall t\in[0,T].\,\xx(t)\in D_q) \, \longrightarrow  \,(\forall t\in [0,T].\,\xx(t)\in P)\enspace .$$
\end{definition}

Intuitively, $P$ is a CI of $(D_q,\fb_q)$ if
any continuous evolution starting from the intersection of $P$
and $D_q$ stays in $P$ as long as it is still in $D_q$.
If $D_q=\mathbb R^n$, then a CI of $(D_q,\fb_q)$ coincides with the standard \emph{(positive) invariant set} (see \cite{Blanchini99}) of the dynamical system defined by $\fb_q$; otherwise if $D_q$ is a proper subset of $\mathbb R^n$, then generally the notion of CI is weaker.

\begin{example}\label{eg:di}
Suppose
$D_q\,\define\,x>0$ and $\fb_q=(-y,x)$. Obviously $P\,\define\,y\geq 0$ is not a positive invariant set of $\fb_q$, whereas $P$ is a CI of $(D_q,\fb_q)$ according to Definition \ref{dfn:DI}. See Fig.~\ref{fig:CI} for an illustration.
\end{example}
\begin{figure}
\begin{center}
\includegraphics[width=1.7in,height=1.7in]{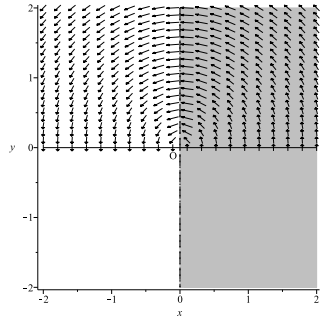}
\end{center}
\caption{Illustration of continuous invariant.}
\label{fig:CI}
\end{figure}

\subsection{The Abstract Synthesis Procedure}\label{sec:abstract-frame}
To solve Problem \ref{dfn:prblm} amounts to refining domains and guards of $\mathcal{H}$ by removing \emph{bad} states from domain $D_{\mathcal{H}}$. A state
$(q, \xx) \in  D_{\mathcal{H}}$ is \emph{bad} if
the hybrid trajectory starting from $(q,\xx)$ either blocks
$\mathcal H$ or violates $S$;
otherwise, it is called a \emph{good} state. From Definition
\ref{dfn:DI}, we observe that the set of good states of $\mathcal H$ can be approximated using continuous invariants, which results in the following solution to Problem \ref{dfn:prblm}.

\begin{theorem}\label{thm:main}
Let $\mathcal H$ and $S$ be the same as in Problem \ref{dfn:prblm}.
Suppose $D_q'$ is a closed subset of $\mathbb R^n$ for all $q\in Q$ and $D_q'\neq \emptyset$ for at least one $q$. If we have
\begin{itemize}
  \item[(c1)] for all $q\in Q$, $D_q'\subseteq D_q\cap S_q$;
  \item[(c2)] for all $q\in Q$, \,$D_q'$ is a continuous invariant of $(H_q,\fb_q)$ with $$H_q\,\define\,\big(\bigcup_{e=(q,q')\in E}G_e'\big)^c\,,$$
\end{itemize}
where $G_e'\,\define\,G_e\cap D_{q'}'$ and $A^c$ denotes the complement of $A$ in $\mathbb R^n$,\,
then the HA $\mathcal H'=(Q,X,f,D',E,G')$ is a solution to
Problem \ref{dfn:prblm}.
\end{theorem}
\begin{proof}
  Please refer to the Appendix \ref{app:a}.\qed
\end{proof}
\begin{remark} Intuitively, by (c1), $D_q'$ is a refinement of $D_q$ and is also contained in the safe region; by (c2), any trajectory starting from $D_q'$ will either stay in $D_q'$ forever, or finally intersect one of the transition guards enabling jumps from $q$ to a certain $q'$, thus guaranteeing satisfaction of the non-blocking requirement.
\end{remark}

Based on Theorem \ref{thm:main}, we give below the steps of a
template-based method for synthesizing a switching controller.
\begin{itemize}
  \item[(s1)] {\bf Template assignment:} assign to each $q\in Q$
    a template parametrically specifying $D_q'$, which can be seen as a refinement of $D_q$ and will be instantiated to be the continuous invariant at $q$;
  \item[(s2)] {\bf Guard refinement:} refine guard $G_e$ for each $e=(q,q')\in E$ by setting $G_e'\,\define\,G_e\cap D_{q'}'$\,;
  \item[(s3)] {\bf Deriving synthesis conditions:} encode
    (c1) and (c2) in Theorem \ref{thm:main} into constraints on parameters appearing in templates;
  \item[(s4)] {\bf Constraint solving:} solve the constraints
    derived from (s3) in terms of the parameters;
  \item[(s5)] {\bf Parameters instantiation:} find an appropriate instantiation of $D_q'$ and $G_e'$ such that $D_q'$ are closed sets for all $q\in Q$, and $D_q'$ is nonempty for at least one $q\in Q$; if such an instantiation is not found, we choose a new set of templates and go back to (s1).
\end{itemize}

\paragraph{Remarks:}
\begin{enumerate}
    \item The implementability of the above method depends on the language used to specify the hybrid system, the safety property, as well as the templates chosen for their refinements. If all appearing expressions are specified using polynomials, the computability of the abstract procedure is guaranteed by Tarski's result \cite{Tarski51}. This will be assumed in the rest of this paper.
    \item In (s3), condition (c1) can be encoded into a first-order polynomial formula straightforwardly; encoding of (c2) into first-order polynomial constraints is based on our previous work in \cite{emsoft11} about a relatively complete method for generating CIs (see Section~\ref{sec:gdi}).
    \item We use \emph{quantifier elimination} (QE) to solve the first-order polynomial constraints obtained in (s4).
    \item The shape of chosen templates in (s1) determines the likelihood of success of the above procedure, as well as the complexity of QE in (s4). In Section \ref{sec:heuristic}, we discuss heuristics for choosing appropriate templates using the \emph{qualitative analysis} discussed in \cite{Kapur97}.
\end{enumerate}

\subsection{A Relatively Complete Method for Generating CIs}\label{sec:gdi}
In \cite{emsoft11} we presented a sound and relatively complete
approach for generating \emph{semi-algebraic} CIs for $(D_q,\fb_q)$ with semi-algebraic $D_q$ and polynomial vector function $\fb_q$. We review here the key ideas; for details the reader can consult \cite{emsoft11}. Below, we drop the subscript corresponding to the mode $q$.

The basic idea can be explained as follows for the simplest case, namely $D\,\define\,h(\xx)>0$ and $P\,\define\,p(\xx)\geq 0$.
Let $\partial P\,\define\,p(\xx)=0$ be the boundary of $P$ and $\xx(t)$ be the continuous evolution of $\fb$ starting from $\xx_0$. It can be shown that $P$ is a CI of $(D,\fb)$ if and only if for all $\xx_0\in \partial P\cap D\,$:
\begin{equation}
\exists\varepsilon >0\,\forall t\in [0,\varepsilon].\,p(\xx(t))\in P  \enspace. \label{eqn:CI-condition}
\end{equation}
Intuitively, (\ref{eqn:CI-condition}) means that trajectories starting at the the boundary of $P$ will stay in $P$ for a small amount of time.

Given that $p$ and $\fb$ are polynomials and thus analytic, the \emph{Taylor expansion} of $p(\xx(t))$ at $t=0$
\begin{eqnarray}
p(\xx(t))& = &p(\xx_0)+\frac{\ud p}{\ud t}\cdot t + \frac{\ud^{2} p}{\ud t^{2}}\cdot\frac{t^{2}}{2!}+\cdots
\label{eqn:taylor1}
\end{eqnarray}
converges in a neighborhood of $0$.

Define the {\itshape
Lie derivatives} of $p$ along $\fb$, $L^n_{\fb} p: \mathbb
R^n\longrightarrow \mathbb R$ for $n\in \mathbb N$, as follows:
\begin{itemize}
\item[$\bullet$] $L^0_{\fb} p(\xx)=p(\xx)$\,;
\item[$\bullet$] $L^n_{\fb} p(\xx)=\big(\frac{\partial}{\partial \xx} L^{n-1}_{\fb} p(\xx),
\fb(\xx)$\big), for $n>0$,
\end{itemize}
where $(\cdot, \cdot)$ is the inner product of two vectors, i.e. $\big(
\,(a_1,\ldots,a_n), ( b_1,\ldots, b_n )\,\big) = \sum_{i=1}^n a_ib_i$.
Using Lie derivatives, (\ref{eqn:taylor1}) rewritten as
\begin{eqnarray}
p(\xx(t))& = & L_{\fb}^0 p(\xx_0)+ L_{\fb}^1 p(\xx_0)\cdot t + L_{\fb}^2 p(\xx_0)\cdot \frac{t^{2}}{2!}+\cdots+L_{\fb}^i p(\xx_0)\frac{t^{i}}{i!}+\cdots \enspace \enspace \label{eqn:taylor2}
\end{eqnarray}

Combining (\ref{eqn:CI-condition}) and (\ref{eqn:taylor2}), our main result of continuous invariant generation (in the simplest case) can be stated as follows.
\begin{theorem}[Necessary and Sufficient Criterion for CIs \cite{emsoft11}]\label{thm:crit}
Given a system $(D, \mathbf{f})$ with $D\, \define \,
h(\xx)> 0$, it has a continuous invariant of
the form $P\,\define\, p(\xx)\geq 0$ if and only if\,
$\forall \xx. \big(p(\xx)=0\wedge h(\xx)>0\longrightarrow \psi(p,\fb)\big)$,
where
\begin{displaymath}
  \psi(p,\mathbf{f}) \,\define \,
 \begin{array}{ll}
     & L^1_{\fb} p(\xx) >0 \\
\vee\, & L^1_{\fb} p(\xx) =0\wedge L^2_{\fb} p(\xx) >0 \\
\vee\, & \cdots\\
\vee\, & L^1_{\fb} p(\xx) =0\wedge\cdots\wedge L^{N_{p,\fb}-1}_{\fb} p(\xx)=0 \wedge L^{N_{p,\fb}}_{\fb} p(\xx)>0 \\
\vee\, & L^1_{\fb} p(\xx) =0\wedge\cdots\wedge L^{N_{p,\fb}-1}_{\fb} p(\xx)=0 \wedge L^{N_{p,\fb}}_{\fb} p(\xx)=0
\end{array}
\end{displaymath}
with $N_{p,\fb}\in\mathbb N$ computed from $p$ and $\fb$.
\end{theorem}
\begin{proof}
  Please refer to \cite{emsoft11}. \qed
\end{proof}

Intuitively, Theorem \ref{thm:crit} means that on the boundary of $P$, up to the $N_{p,\fb}$-th order, the first non-zero higher order Lie derivative of $p$ w.r.t $\fb$ is non-negative.

The above theorem can be generalized for parametric polynomials $p(\uu,\xx)$, thus enabling us to use polynomial templates and QE to automatically discover CIs. Such a method for CI generation is relatively complete, that is, if there exists a CI in the form of the predefined template, then we are able to find one.

\begin{example}
Consider the system $(\mathbb R^2, \fb)$ from \cite{Taly09} with
$\fb\,\define\,(\dot x=1-y,\,\dot y=x)$, which has a continuous invariant $p\geq 0$ with $p\,\define\,-(-x^2-y^2+2y)^2$,  defining the circumference of a circle.

In \cite{Taly09}, sound and complete inference rules are
given for invariants that are linear, quadratic, smooth or
convex. However, it was pointed out in \cite{Taly09} that all these rules failed to prove the invariance property of $p\geq 0$, as $p$ is not linear or quadratic, nor is it smooth or convex.
Furthermore, by a simple computation we get $L_{\fb}^k p\equiv 0$ for all $k\geq 1$, so the sound but incomplete rule in \cite{Taly09,Taly11} which involves only strict inequalities over finite-order Lie derivatives is also inapplicable. However, from $L_{\fb}^1 p\equiv 0$ we get
$N_{p,\fb}=0$, and then according to Theorem \ref{thm:crit}, $p\geq 0$ can be verified since $\forall x\forall y.\,\left(-(-x^2-y^2+2y)^2=0\longrightarrow \textsf{true}\right)$
holds trivially.

Although the rule in \cite{PlatzerClarke08} can also be used to check the invariant $p\geq 0$, generally it only works on very restricted invariants. Even for linear systems like $(\mathbb R, \dot x=x)$, it cannot prove the invariant $x\geq 0$ because $\forall x.x\geq 0$ is obviously {false}, while our approach requires $\forall x. (x=0 \rightarrow \textsf{true})$ which is trivially true.

The above examples show the generality and flexibility of our approach, using
which it is possible to generate CIs in many
general cases, and hence gives more possibility to synthesize a controller based on our understandings of the kind of controllers that can be synthesized using methods in \cite{Taly11,Taly-10-reach,Sturm-Tiwari-issac11}.
\end{example}

\subsection{Heuristics for Predefining Templates}\label{sec:heuristic}
The key steps of the qualitative analysis used in \cite{Kapur97}
are as follows.
\begin{enumerate}
\item The evolution behavior (increasing or decreasing) of continuous variables in each mode
    is inferred from the differential equations (using first
    or second order derivatives);
\item  \emph{control critical} modes, at which the maximal (or minimal) value of a continuous variable is achieved, can be identified;
\item the safety requirement is imposed to obtain constraints on guards of transitions leading to control critical modes, and
\item then this information is propagated to other modes.
\end{enumerate}

Next, we illustrate how such an analysis helps
in predefining templates for the running example.
\begin{example}[Nuclear Reactor Temperature Control]\label{eg:nuclear}
Our goal is to synthesize a switching  controller for the system in Example \ref{ex:nuclear0} with the global safety requirement that the temperature of the core lies between  $510$ and $550$, i.e. $S_i\,\define\,510\leq x \leq 550$ for $i=1,2,3,4$.
\end{example}
\begin{itemize}
  \item[1)] {\bf Refine domains.} \,Using the safety requirement,
     domains $D_i$ for $i=1,2,3,4$ are refined by $D_i^s\,\define\,D_i\cap S_i$, e.g. $D_1^s\,\define\, p=0\wedge 510\leq x\leq 550$\,.
  \item[2)] {\bf Infer continuous evolutions.} Let $l_1\,\define\,x/10-6p-50=0$ be the \emph{zero-level} set of $\dot x$ and check how $x$ and $p$ evolve in each mode. For example, in $D_2^s$, $\dot x>0$ on the left of $l_1$ and $\dot x<0$ on the right; since $p$ increases from $0$ to $1$, $x$ first increases then decreases and achieves maximal value when crossing $l_1$.
  \item[3)] {\bf Identify critical control modes.} By 2), $q_2$ and $q_4$ are critical control modes at which the evolution direction of $x$ changes.
  \item[4)] {\bf Generate control points.} By 3), we can get a control point $E(5/6,550)$ at $q_2$ by taking the intersection of $l_1$ and the safety upper bound $x=550$; and $F(1/6,510)$ can be obtained similarly at $q_4$.
  \item[5)] {\bf Propagate control points.} $E$ is backward propagated to $A(0,a)$ using trajectory $\scriptstyle{\wideparen{AE}}$ at $q_2$, and then to $C(1,c)$ using trajectory $\scriptstyle{\wideparen{CA}}$ at $q_4$; similarly, by propagating $F$ we get $D$ and $B$. (See Fig.~\ref{fig:ci-templates}.)
  \item[6)] {\bf Construct templates.} For brevity, we only show how to construct $D_2'$. Intuitively, $p=0$, $p=1$, $\scriptstyle{\wideparen{AE}}$ and $\scriptstyle{\wideparen{BD}}$ forms the boundaries of $D_2'$. In order to get a semi-algebraic template, we need to fit $\scriptstyle{\wideparen{AE}}$ and $\scriptstyle{\wideparen{BD}}$ by polynomials using points $A,E$ and $B,D$ respectively. By 2), $\scriptstyle{\wideparen{AE}}$ has only one extreme point $E$ in $D_2^s$ and is tangential to $x=550$ at $E$.
      The simplest algebraic curve that can exhibit a shape similar to $\scriptstyle{\wideparen{AE}}$ is the parabola through $A,E$ opening downward with $l_2\,\define\,p=\frac{5}{6}$ the axis of symmetry.
      Therefore to minimize the degree of terms appearing in templates, we do not resort to polynomials with degree greater than 3. This parabola can be computed using the coordinates of $A,E$ as: $x-550-\frac{36}{25}(a-550)(p-\frac{5}{6})^2= 0$\,.
\end{itemize}
\begin{figure}
\begin{center}
\includegraphics[width=1.8in,height=1.5in]{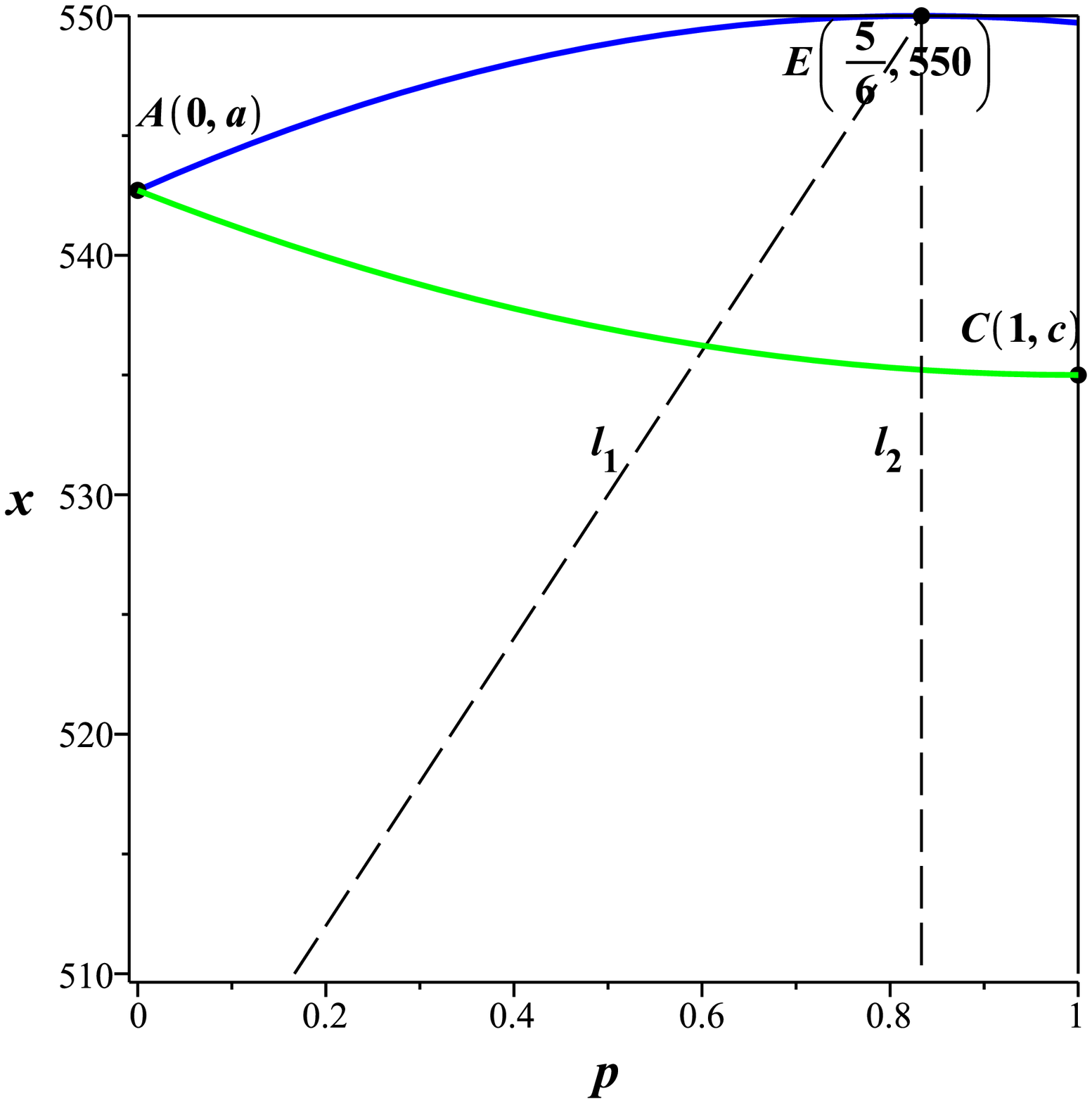}
\hspace{1cm}
\includegraphics[width=1.8in,height=1.5in]{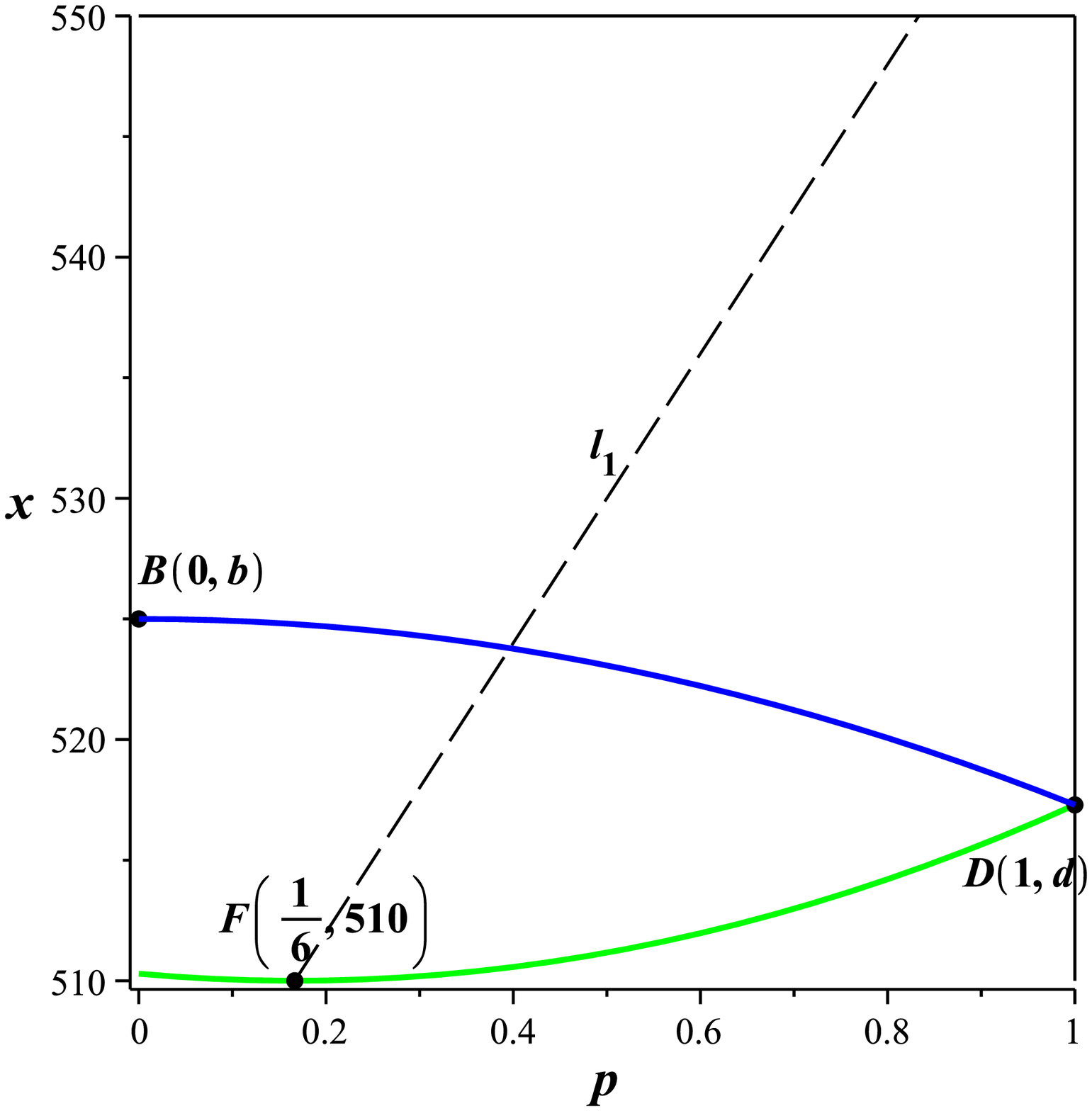}
\end{center}
\caption{Control points propagation.}\label{fig:ci-templates}
\end{figure}

Through the above analysis, we generate the following templates:
\begin{itemize}
  \item[$\bullet$] $D_1'\,\define\,\, p=0 \,\wedge\, 510\leq x\leq a$\,;
  \item[$\bullet$] $D_2'\,\define\,\, 0\leq p\leq 1 \,\wedge\, x-b\geq p(d-b) \,\wedge\, x-550-\frac{36}{25}(a-550)(p-\frac{5}{6})^2\leq 0 $\,;
  \item[$\bullet$] $D_3'\,\define\,\, p=1 \,\wedge\, d \leq x\leq 550$\,;
  \item[$\bullet$] $D_4'\,\define\,\, 0\leq p\leq 1 \,\wedge\, x-a\leq p(c-a) \,\wedge\, x-510-\frac{36}{25}(d-510)(p-\frac{1}{6})^2\geq 0 $\,,
\end{itemize}
\noindent
in which $a,b,c,d$ are parameters.
Note that without qualitative analysis, a single generic \emph{quadratic} polynomial over $p$ and $x$ would require ${{2+2} \choose 2}= 6$ parameters.

The above heuristics works well on planar systems and can also be applied to three-dimensional systems. We are further generalizing the
heuristics to cover a wider class of hybrid automata.

Based on the framework presented in Section \ref{sec:abstract-frame},
we show below how to synthesize a switching controller for the system in
Example \ref{eg:nuclear} step by step.
\begin{example}[Nuclear Reactor Temperature Control Contd.]
\begin{itemize}
\item[(s1)]
The four invariant templates are defined in Section~\ref{sec:heuristic}.
\item[(s2)] The four guards are refined by setting $G_{ij}'\,\define\,G_{ij}\cap D_{j}'$\,:
\begin{itemize}
  \item $G_{12}'\,\define\,\,p=0\,\wedge \,b\leq x\leq a$\,;
  \item $G_{23}'\,\define\,\,p=1\,\wedge \,d\leq x\leq 550$\,;
  \item $G_{34}'\,\define\,\,p=1\,\wedge\,d\leq x\leq c$\,;
  \item $G_{41}'\,\define\,\,p=0\,\wedge\,510\leq x\leq a$\,.
  \vspace{.1cm}
\end{itemize}
\item[(s3)] Using $D_i'$ and $G_{ij}'$ we can derive the synthesis condition, which is a first-order polynomial formula in the form of $\phi\,\define\,\forall x\forall p. \varphi(a,b,c,d,x,p)$. We do  not include $\phi$ here due to its big size.
\item[(s4)] By applying QE to $\phi$ we get the following solution to the parameters:
    \begin{equation}\label{eqn:res-nuclear}
    a=\frac{6575}{12} \,\wedge\,b= \frac{4135}{8}\,\wedge\, c= \frac{4345}{8}\,\wedge\,d=\frac{6145}{12}\enspace.
    \end{equation}
\item[(s5)] Instantiate $D_i'$ and $G_{ij}'$ by (\ref{eqn:res-nuclear}). It is obvious that all $D_i'$ are nonempty closed sets. According to Theorem \ref{thm:main}, we get a safe switching controller for the nuclear reactor system. The left picture in Fig.~\ref{fig:shape-inv} is an illustration of $D_2'$.
\end{itemize}
\begin{figure}
\begin{center}
\includegraphics[width=1.5in,height=1.5in]{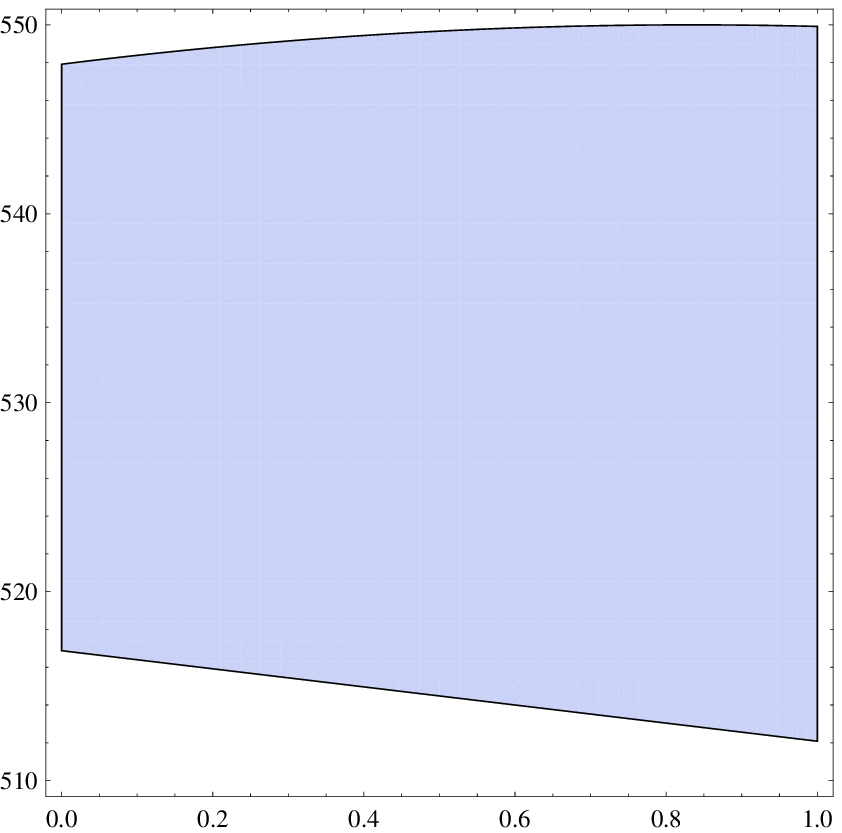}
\hspace{.1cm}
\includegraphics[width=1.5in,height=1.5in]{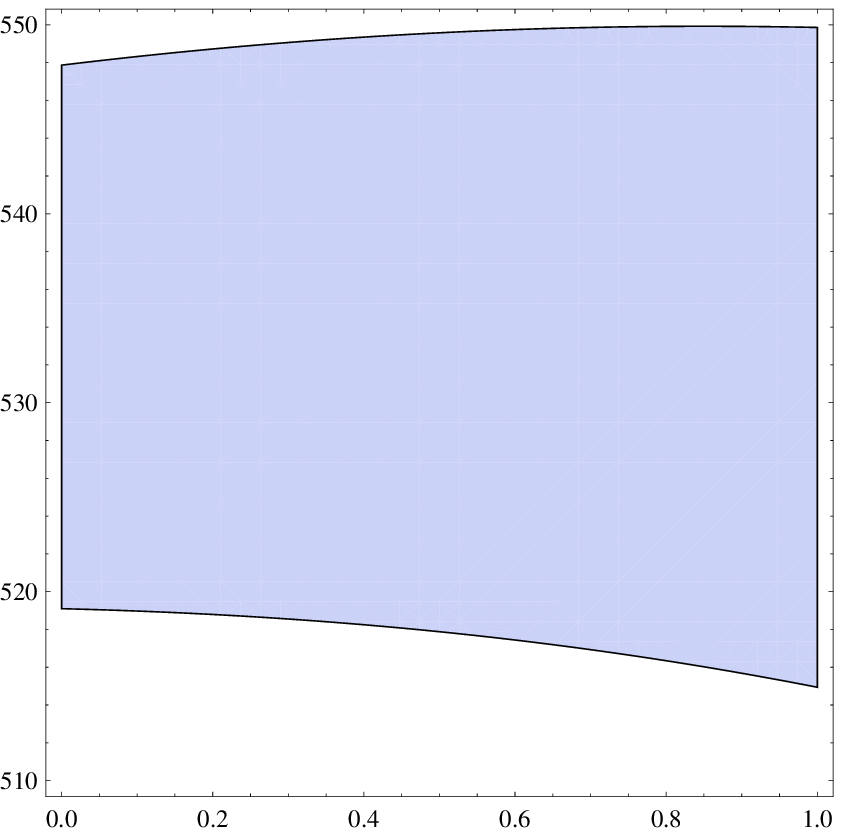}
\hspace{.1cm}
\includegraphics[width=1.5in,height=1.5in]{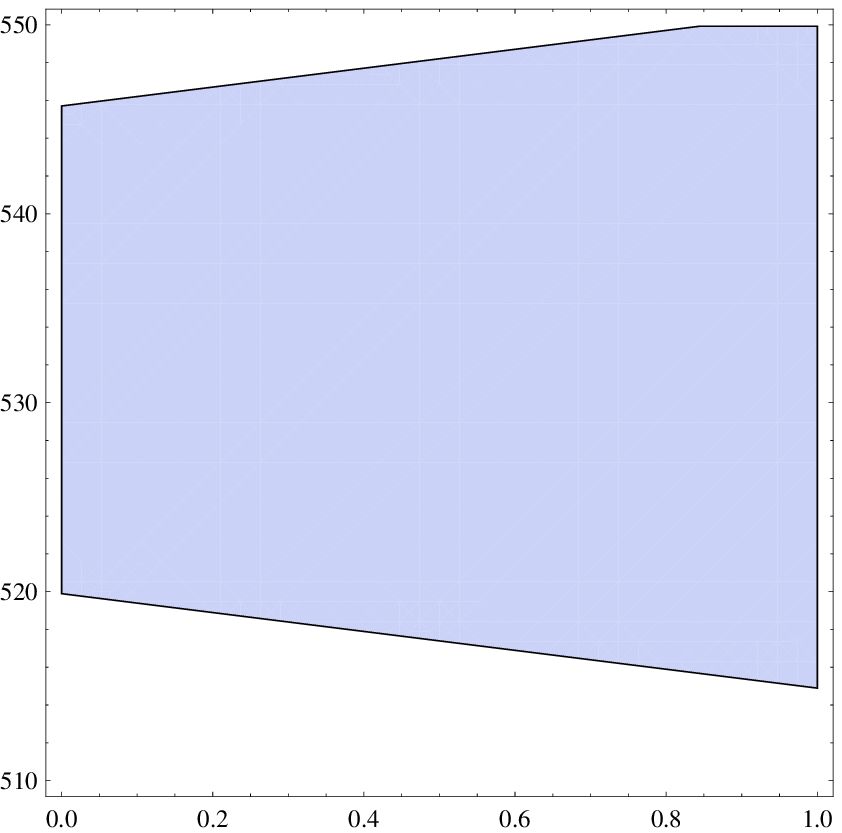}
\end{center}
\caption{Shape of synthesized continuous invariants.}\label{fig:shape-inv}
\end{figure}

In \cite{Kapur97}, an upper bound $x=547.97$ for $G_{12}$ and a lower bound $x=512.03$ for $G_{34}$ are obtained by solving the differential equations at mode $q_2$ and $q_4$ respectively.
By (\ref{eqn:res-nuclear}), the corresponding bounds generated
here are $x\leq \frac{6575}{12}=547.92$ and $x\geq
\frac{6145}{12}=512.08$.

As should be evident from the above results, in contrast to
\cite{Kapur97}, where differential equations are solved to get
closed-form solutions, we are able to get good approximate
results without requiring closed-form solutions. This indicates
that our approach should work well for hybrid automata where
differential equations for modes need not have closed form solutions.
\end{example}

\section{Synthesis by Generating CIs Numerically}
The QE-based approach crucially depends upon quantifier elimination techniques. It is well known that the complexity of a general purpose QE method over the full theory of real-closed fields
is \emph{doubly exponential} in the number of variables \cite{dh88}. Therefore it is desirable to develop heuristics to do QE more efficiently. As shown in Section \ref{sec:heuristic}, qualitative analysis helps in reducing the number of parameters in templates. Another possible way to address the issue of high computational cost is resorting to numerical methods. In this section, we will discuss the application of two such approaches to the nuclear reactor example.

\subsection{The SOS Relaxation Approach}
Let $\mathbb R[x_1,x_2,\ldots,x_n]$, or $\mathbb R[\xx]$ for short, denote the polynomial ring over variables $x_1,x_2,\ldots,x_n$ with real coefficients. A \emph{monomial} is a special polynomial in the form of $x_1^{\alpha_1}x_2^{\alpha_2}\cdots x_n^{\alpha_n}$ with $(\alpha_1,\alpha_2,\ldots,\alpha_n)\in\mathbb N^n$. Any polynomial $p(\xx)\in \mathbb R[\xx]$ of degree $d$ can be written as a linear combination of $n+d \choose d$ monomials, i.e.
$$p(\xx) = \sum_{\alpha_1+\alpha_2+\cdots+\alpha_n\leq d} c_{(\alpha_1,\alpha_2,\ldots,\alpha_n)} \cdot x_1^{\alpha_1}x_2^{\alpha_2}\cdots x_n^{\alpha_n}\enspace\,.$$
A polynomial $p$ is called an SOS (sum-of-squares) if there exist $s$ polynomials $q_1,q_2,\ldots,q_s$ s.t.
$$p=\sum_{1\leq i\leq s} q_i^2\enspace .$$
It is obvious that any SOS $p$ is non-negative, i.e. $\forall \xx\in\mathbb R^n. \,p(\xx)\geq 0\enspace .$

The basic idea of SOS relaxation is as follows: to prove that a polynomial $p$ is nonnegative, we can try to show that $p$ can be decomposed into a sum of squares, a trivially sufficient condition for non-negativity (but generally not necessary); similarly, to prove $p\geq 0$ on the semi-algebraic set $q\geq 0$, it is sufficient to find two SOS $r_1,r_2$ such that
$p=r_1+r_2\cdot q$.

SOS relaxation is attractive because the searching for SOS decomposition can be reduced to a semi-definite programming (SDP) problem according to the following equivalence \cite{Parrilo-thesis}:
\begin{quote}
  A polynomial $p$ of degree $2d$ is an SOS if and only if there exists a semi-definite matrix $Q$ such that $p=\mathbf q\cdot Q \mathbf\cdot \mathbf q^{T}$, where $\mathbf q$ is a ${n+d \choose d}$-dimensional row vector of monomials with degree $\leq d$.
\end{quote}
SDP is a convex programming that is solvable in \emph{polynomial} time using numerical methods such as the interior point method \cite{Boyd96}. Therefore the SOS computation is a tractable problem.

We now show how SOS can be related to CI generation. Let $p\geq 0$ be a parametric template defined for the system $(h>0,\fb)$. By Theorem \ref{thm:crit}, a sufficient condition for $p\geq 0$ to be a CI of $(h>0,\fb)$ is
$$\forall \xx. \big(p(\xx)=0\wedge h(\xx)>0\longrightarrow L_{\fb}^1 p (\xx)>0\big)\,,$$
which can be further strengthened to
\begin{equation}\label{eqn:relaxed-ci}
\forall \xx. \big(h(\xx)>0\longrightarrow L_{\fb}^1 p (\xx)>0\big)\,.
\end{equation}
Using SOS relaxation, a sufficient condition for (\ref{eqn:relaxed-ci})
is
\begin{equation}\label{eqn:sos-relax}
  L^1_{\fb}p= s_1+s_2\cdot h + \varepsilon\enspace,
\end{equation}
where $s_1,s_2$ are SOS and $\varepsilon$ is a positive constant. Solve (\ref{eqn:sos-relax}) for the parameters appearing in $p$ and then we can get a CI $p\geq 0$ of $(h>0,\fb)$.

For the nuclear reactor example, we define two general \emph{quartic} templates
$$p\geq 0\wedge p\leq 1\wedge \sum_{\alpha_1+\alpha_2\leq 4} c_{(\alpha_1,\alpha_2)}\cdot p^{\alpha_1}x^{\alpha_2}\leq 0$$
and
$$p\geq 0\wedge p\leq 1\wedge \sum_{\alpha_1+\alpha_2\leq 4} d_{(\alpha_1,\alpha_2)}\cdot p^{\alpha_1}x^{\alpha_2}\leq 0$$
for mode $q_2$ and $q_4$ respectively. Using the SOS relaxation techniques discussed above, the following refined domains are obtained:
\begin{itemize}
  \item[$\bullet$] $D_1'\,\define\,\, p=0 \,\wedge\, 510\leq x\leq 547.85$\,;
  \item[$\bullet$] $D_2'\,\define\,\, 0\leq p\leq 1 \,\wedge\, 34468.9 - 9941.89 p + 1114.38 p^2 + 67.3261 p^3 + 0.925 p^4 - 129.316 x + 37.7294 p x - 4.9669 p^2 x - 0.1303 p^3 x + 0.1212 x^2 - 0.0358 p x^2 + 0.0054 p^2 x^2 \leq 0$\,;
  \item[$\bullet$] $D_3'\,\define\,\, p=1 \,\wedge\, 512.09 \leq x\leq 550$\,;
  \item[$\bullet$] $D_4'\,\define\,\, 0\leq p\leq 1 \,\wedge\, 46082.8 + 8787.98 p + 1473.0 p^2 + 93.8933 p^3 + 1.635 p^4 - 174.456 x - 34.1747 p x - 4.7397 p^2 x - 0.1829 p^3 x + 0.1649 x^2 +  0.0332 p x^2 + 0.0037 p^2 x^2 \leq 0$\,.
\end{itemize}
The picture in the middle of Fig.~\ref{fig:shape-inv} illustrates the synthesized $D_2'$.

\subsection{The Template Polyhedra Approach}
Polyhedral sets are a popular family of (positive) invariants of linear (continuous or discrete) systems \cite{Blanchini99}.
A \emph{convex polyhedron} in $\mathbb R^n$ can be represented using linear inequality constraints as $Q \xx \leq \mathbf \rho$, where $Q\in \mathbb R^{r\times n}$ is an $r\times n$ matrix, and $\xx\in\mathbb R^{n\times1}, \mathbf \rho\in\mathbb R^{r\times 1}$ are column vectors.

Given a linear continuous dynamical system $\dot \xx=A\xx$ with $A\in\mathbb R^{n\times n}$, the following result about (positive) polyhedral invariant set is established in \cite{Castelan93}.
\begin{proposition}\label{prop:polyhera-pinv}
  The polyhedron $Q \xx \leq \mathbf \rho$ is a positive invariant set of $\dot \xx=A \xx$ if and only if there exists an essentially non-negative{\footnote{A square matrix is \emph{essentially non-negative} if all its entries are non-negative except for those on the diagonal. Besides, given a matrix $M$, in this paper the notation $M\geq 0$, $M>0$ and $M=0$ should be interpreted entry-wisely.}} matrix $H\in \mathbb R^{r\times r}$ satisfying $HQ=QA$ and $H\rho\leq 0$.
\end{proposition}

By simply applying the famous \emph{Farkas' lemma} \cite{Tiwari08}, we can generalize Proposition \ref{prop:polyhera-pinv} and give a sufficient condition for polyhedral CIs of linear dynamics with open polyhedral domain.

\begin{proposition} \label{prop:my-polyhedra}
Let $\fb\,\define\,A\xx + \mathbf b$ and $D\,\define\,\mathbf c\mathbf x < a$, where $a\in\mathbb R$, $\mathbf b\in \mathbb R^{n\times 1}$ is a column vector, and $\mathbf c\in \mathbb R^{1\times n}$ is a row vector. Then the polyhedron $Q \xx \leq \mathbf \rho$ is a CI of the system
$(D,\fb)$ if there exist essentially non-negative matrix $H\in\mathbb R^{r\times r}$, and non-negative column vectors $\eta\geq 0,\xi\geq 0,\lambda\geq 0$ in $\mathbb R^{r\times 1}$ such that
\begin{itemize}
  \item[(1)] $HQ + \xi \mathbf c - {\sf diag}(\lambda)QA = 0$\enspace ;
  \item[(2)] $H\rho + \eta + \xi a + {\sf diag}(\lambda)Q\mathbf b = 0$ \enspace ;
  \item[(3)] $\xi+\eta>0$\enspace ,
\end{itemize}
where {\sf diag}$(\lambda)$ denotes the $r\times r$ diagonal matrix with the main diagonal $\lambda$.
\end{proposition}
\begin{proof}
  Please refer to Appendix \ref{app:b}. \qed
\end{proof}

Proposition \ref{prop:my-polyhedra} serves as the basis of automatic generation of polyhedral CIs for linear systems. To reduce the number of parameters in a polyhedral template, we propose the use of \emph{template polyhedra}. The idea is to partly fix the shape of the invariant polyhedra by fixing the orientation of their facets. Formally, a template polyhedron is of the form $Q\xx\leq \rho$ where $Q$ is fixed a priori and $\rho$ is to be determined. Any instantiation of $\rho$ from $\mathbb R^{r\times 1}$ produces a concrete polyhedron. In this paper, since the system is planar, we choose $Q$ in such a way that its row vectors form a set of uniformly distributed directions on a unit circle, i.e.
$$\mathbf q_i=\big(\cos(\frac{i-1}{r}2\pi), \sin(\frac{i-1}{r}2\pi)\big)$$ for $1\leq i\leq r$, where $\mathbf q_i$ denotes the $i$-th row of $Q$. It is easy to see that $Q\xx\leq \rho$ is just a rectangle when $r=4$, and an octagon when $r=8$.

In order to determine $\rho$, we have to solve the constraints derived from Proposition \ref{prop:my-polyhedra}. Note that since both $H$ and $\rho$ are indeterminate, the constraint (2) becomes \emph{bilinear}, making the problem NP-hard \cite{VanAntwerp2000} to solve. It is however still more tractable using modern BMI (bilinear matrix inequality) solvers compared to QE. The details of applying numerical solvers will be discussed in the Conclusion part.

Using octagonal templates\footnote{To reduce the number of facets needed in the template, we scaled the variable $x$ by a factor of $0.2$, i.e. let $x=5x'$, and rescaled the generated invariants by $5$.} for mode $q_2$ and $q_4$ in the nuclear reactor example, we obtain the following refined domains.
\begin{itemize}
  \item[$\bullet$] $D_1'\,\define\,\, p=0 \,\wedge\, 510\leq x\leq 545.50$\,;
  \item[$\bullet$] $D_2'\,\define\,\,Q (p,x)^T\leq \rho_1$ with
    \begin{displaymath}
      Q =
      \scriptsize{
      \left( \begin{array}{cccccccc}
5.0000 \,& 3.5355 \,& 0.0000 \,& -3.5355\, & -5.0000 \,& -3.5355\, & -0.0000 \,& 3.5355\\
0.0000 & 0.7071 & 1.0000 & 0.7071  & 0.0000 & -0.7071 & -1.0000 & -0.7071
     \end{array} \right)^T
     }
    \end{displaymath}
    and
      \begin{displaymath}
      \rho_1 =
      \scriptsize{
      \left( \begin{array}{cccccccc}
      5.0000 \,& 392.4429 \,& 549.9276 \,& 385.8688\, & 0.0000\, & -367.6169\, & -514.5244\, & -360.2745
     \end{array} \right)^T
     } \,;
    \end{displaymath}
  \item[$\bullet$] $D_3'\,\define\,\, p=1 \,\wedge\, 514.50 \leq x\leq 550$\,;
  \item[$\bullet$] $D_4'\,\define\,\,Q (p,x)^T\leq \rho_2$ with the $Q$ in $D_2'$ and
      \begin{displaymath}
      \rho_2 =
      \scriptsize{
      \left( \begin{array}{cccccccc}
      5.0000\, & 384.1267\, & 545.5548 \,& 384.7484\, & 0.0000\, & -360.6299\, & -510.1431\, & -360.1948
     \end{array} \right)^T
     } \,.
    \end{displaymath}
\end{itemize}
In Fig.~\ref{fig:shape-inv}, the picture on the right  illustrates the synthesized $D_2'$.

\section{Conclusion and Discussion}
We have extended a template-based approach for
synthesizing switching controllers for semi-algebraic hybrid
systems by combining symbolic invariant generation methods using quantifier elimination with qualitative methods to determine the likely shape of invariants. We have also investigated the application of numerical methods to gain high level of scalability and automation. A summary comparison of the three proposed approaches, i.e. the QE-based, SOS-relaxation and template-polyhedra approaches, can be given in the following aspects.

\begin{itemize}
  \item {\bf Applicability:} the QE-based approach can be applied to any semi-algebraic system; SOS relaxation techniques can be applied to semi-algebraic systems for which SOS encoding is possible; the template-polyhedra approach is only applicable to linear systems.
  \item {\bf Design of Templates:} the QE-based approach demands much heuristics in determining templates, while the other two need little human effort.
  \item {\bf Relative Completeness:} only the QE-based approach is relative complete w.r.t. the predefined family of templates, but we believe that the template-polyhedra approach can be made relatively complete by improving Prop.~\ref{prop:my-polyhedra}.
  \item {\bf Quality of Controllers:} the QE-based and SOS-relaxation approaches can generate arbitrary (non-convex) semi-algebraic invariants, while the template-polyhedra approach can only generate convex polyhedral invariants; for the nuclear reactor example, we can see from Fig.~\ref{fig:shape-inv} that the QE-based approach produced larger refined domains and transition guards, but such superiority is not a necessity and relies greatly on the quality of heuristics.
  \item {\bf Computational Cost:}  for the QE-based approach, we have used the algebraic tools Redlog \cite{Redlog} and QEPCAD (the {\sf slfq} function) \cite{qepcad} to perform QE; for the numerical approaches, we use the MATLAB optimization toolbox YALMIP \cite{YALMIP,Lofberg2009} as a high-level modeling environment and the interfaced external solvers SeDuMi \cite{SeDuMi} and PENBMI \cite{Penbmi} (the TOMLAB \cite{Tomlab} version) to solve the underlying SDP and BMI problems respectively. Table \ref{tbl:timing} shows the time cost of three approaches applied to the nuclear reactor (NR) example as well as a thermostat (TS) example from \cite{Jha09}. All computations are done on a desktop with a 2.66\,GHz CPU and 4\,GB memory.
\begin{table}
\begin{center}
  \caption{Templates and time cost of three controller synthesis approaches.}\label{tbl:timing}
\begin{tabular}{|c|c|c|c|c|}
  \hline
  \multicolumn{2}{|c|}{\,Approach\,} &\, QE-based\, &\, SOS-relaxation \,& \,template-polyhedra \,\\
  \hline
  \multicolumn{2}{|c|}{Tool} & \,Redlog + {\sf slfq }\, & \,YALMIP + SeDuMi \, & \,YALMIP + PENBMI\,\\
  \hline
   \multirow{2}{*}{\,Template\,} & \,NR\, &\, quadratic, {\scriptsize \textsf{\#PARMS}}$\,=4$ \, & generic quartic & 8 facets \\
  \cline{2-5}
  & TS & quadratic, {\scriptsize \textsf{\#PARMS}}$\,=2$ & generic quartic & 10/12 facets\\
  \hline
  \,Time\, & \,NR\, & 12.663 & 1.969 & 0.578\\
  \cline{2-5}
  (sec) & TS & 7.092 & 1.609 & 1.437\\
  \hline
\end{tabular}
\end{center}
\end{table}
We can see that for these two examples the QE-based approach is consistently more expensive in time compared to numerical approaches.
  \item {\bf Soundness:} the QE-based approach is exact while the other two approaches suffer from numerical errors which would cause the synthesis of unsafe controllers. The justification for use of numerical methods is that verification is much easier than synthesis. For example, we have verified posteriorly and symbolically the controllers synthesized by both numerical approaches in this paper. We could also directly encode some tolerance of numerical errors into the synthesis constraints to increase robustness and reduce the risk of synthesizing bad controllers.
\end{itemize}

Our analysis of a nuclear reactor example suggests the effectiveness of all three proposed approaches. We are currently experimenting with these (and more other) methods on more complex examples. We believe that there exists no single method that can solve all the problems. A practical way is to select the most suitable one(s) for any specific problem.

Although the focus of this paper is on the switching controller synthesis problem subject to safety requirements, we plan to extend the proposed approach for reachability and/or optimality requirements as well, by incorporating our previous results on \emph{asymptotic stability} analysis \cite{mics2012} and a case study in optimal control \cite{fm2012}.



\subsubsection{Acknowledgements.} We thank Dr. Jiang Liu for his contribution to our previous joint work on invariant generation.
We also thank Dr.~Matthias Horbach and Mr.~ThanhVu Nguyen for their valuable comments on our old drafts.


\newpage
\appendix
\section{Proof of Theorem \ref{thm:main}}\label{app:a}
We need the following definitions \cite{Sastry00} to prove Theorem \ref{thm:main}.

\begin{definition}[Hybrid Time Set] A hybrid time set is a sequence of intervals $\tau=\{I_i\}_{i=0}^N$ ($N$ can be $\infty$) such that:
\begin{itemize}
  \item[$\bullet$]$I_i=[\tau_i,\tau_i']$ with $\tau_i\leq\tau_i'=\tau_{i+1}$ for all $i<N$;
  \item[$\bullet$]if $N<\infty$, then $I_N=[\tau_N,\tau_N'\rg$ is a right-closed or right-open nonempty interval ($\tau_N'$ may be $\infty$);
\item[$\bullet$]$\tau_0=0$\enspace .
\end{itemize}
\end{definition}

Given a hybrid time set, let $\<\tau\rg=N$ and $\|\tau\|=\sum_{i=0}^N(\tau_i'-\tau_i)$\,.

\begin{definition}[Hybrid Trajectory]\label{dfn:HT}
A hybrid trajectory of $\mathcal{H}$ starting from an initial point $(q_0,\xx_0)\in D_{\mathcal H}$ is a triple $\omega=(\tau,\alpha,\beta)$, where $\tau=\{I_i\}_{i=0}^N$ is a hybrid time set, and $\alpha=\{\alpha_i:I_i\rightarrow Q\}_{i=0}^N$ and $\beta=\{\beta_i:I_i\rightarrow \mathbb R^n\}_{i=0}^N$ are two sequences of functions satisfying:
\begin{enumerate}
  \item Initial condition: $\alpha_0[0]=q_0$ and $\beta_0[0]=\xx_0$;
  \item Discrete transition: for all $i<\<\tau\rg$,  $e=\big(\alpha_i(\tau_i'),\alpha_{i+1}(\tau_{i+1})\big)\in E$, $\beta_i(\tau_i')\in G_e$ and $\beta_{i+1}(\tau_{i+1})=\beta_i(\tau_i')$;
  \item Continuous evolution: for all $i\leq \<\tau\rg$ with $\tau_i<\tau_i'$, if $q=\alpha_i(\tau_i)$, then
      \begin{itemize}
        \item[(1)]for all $t\in I_i$, $\alpha_i(t)=q$,
        \item[(2)]$\beta_i(t)$ is the solution to the differential equation $\dot {\xx}=\fb_q(\xx)$ over $I_i$ starting from $\beta_i(\tau_i)$, and
        \item[(3)]for all $t\in[\tau_i,\tau_i')$, $\beta_i(t)\in D_q$\enspace.
      \end{itemize}
\end{enumerate}
\end{definition}

\subsection*{Proof of Theorem \ref{thm:main}}
\begin{proof} We prove that the three requirements in Problem \ref{dfn:prblm} are satisfied by $\mathcal H'$.

\vspace{.1cm}
\noindent {\bf (r1)} \quad By (c1), we get $D_q'\subseteq D_q$ for all $q\in Q$; by the definition of $G_e'$, $G_e'\subseteq G_e$ for all $e\in E$.
\vspace{.2cm}

\noindent {\bf(r2)}\quad Suppose $\omega=(\tau,\alpha,\beta)$ is a hybrid trajectory starting from $(q_0,\xx_0)\in D_{\mathcal H'}$. We prove
      \begin{equation}\label{eq:r2}
        \forall i\leq \<\tau\rg\,\forall t\in I_i.\, \beta_i(t)\in S_{\alpha_i(t)}
      \end{equation}
      by induction on $\<\tau\rg$.

      If $\<\tau\rg=0$, then $I_0=[0,T_0\rg$ is a right-open or right-closed interval for some $T_0\geq 0$. If $T_0=0$ then $I_0=\{0\}$. By (c1) and condition 1 of Definition \ref{dfn:HT} we have $\beta_0(0)=\xx_0\in D_{q_0}'\subseteq S_{q_0} =S_{\alpha_0(0)}$. If $T_0>0$, by condition 1 and 3 of Definition \ref{dfn:HT} as well as (c1), we have for all $t\in [0,T_0)$, $\beta_0(t)\in D_{q_0}'\subseteq S_{q_0} =S_{\alpha_0(t)}$. If $I_0\neq [0,T_0)$, i.e. $I_0=[0,T_0]$, by noticing that $D_{q_0}'$ is a closed set and $\beta_0(t)$ is continuous over $I_0$, we get $\beta_0(T_0)\in D_{q_0}'\subseteq S_{q_0}= S_{\alpha_0(T_0)}$. Thus we have proved in all cases, $\forall t\in I_0.\, \beta_0(t)\in S_{\alpha_0(t)}$.

      Assume (\ref{eq:r2}) holds for $\<\tau\rg=k\geq 0$. When $\<\tau\rg=k+1$, by assumption, for all $i\leq k$ and all $t\in I_i$ we have $\beta_i(t)\in S_{\alpha_i(t)}$. By condition 2 in Definition \ref{dfn:HT}, there exists $e=(q_k,q_{k+1})\in E$ such that $\alpha_k(\tau_k')=q_k$, $\alpha_{k+1}(\tau_{k+1})=q_{k+1}$ and
      $\beta_{k+1}(\tau_{k+1})=\beta_k(\tau_k')\in G_e'\subseteq D_{q_{k+1}}'$. Consider the hybrid trajectory starting from  $\big(q_{k+1},\beta_{k+1}(\tau_{k+1})\big)\in D_{\mathcal H'}$.
      By applying the same analysis we do for case $\<\tau\rg=0$, we can get $\beta_{k+1}(t)\in S_{\alpha_{k+1}(t)}$ for all $t\in I_{k+1}$. Thus we have proved that (\ref{eq:r2}) holds for $\<\tau\rg=k+1$. Then by induction, (\ref{eq:r2}) holds for all $\<\tau\rg$.
\vspace{.3cm}

\noindent {\bf (r3)} Given $(q_0,\xx_0)\in D_{\mathcal H'}$ (hence $D_{q_0}'\neq \emptyset$), we will construct a non-blocking  hybrid trajectory starting from $(q_0,\xx_0)$.

Suppose $\xx(t)$ is the continuous evolution defined by $\fb_{q_0}$ starting from $\xx_0$. Let $T_{\max}$ be the maximal positive $T$ satisfying \begin{equation}\label{eqn:maximumt}
\xx(t)\in D_{q_0}'\cap H_{q_0}\,\, \mbox{for all}\,\, t\in [0,T)\,,
\end{equation}
if such $T$ exists, and $T_{\max}=0$ otherwise.

If $T_{\max}=\infty$, then we already get an infinite hybrid trajectory starting from $(q_0,\xx_0)$.

If $T_{\max}<\infty$, then by the completeness of $\fb_{q_0}$, $\xx(t)$ must exist on  $[0,T_{\max}+\varepsilon]$ for some $\varepsilon>0$. We assert that
\begin{equation}\label{eq:r3}
\xx(T_{\max})\in (H_{q_0})^c=\bigcup_{e=(q_0,q')\in E}G_e'\enspace .
\end{equation}
If not, i.e. $\xx(T_{\max})\in H_{q_0}$, then there exists $0<\varepsilon'<\varepsilon$ s.t. $\xx(t)\in H_{q_0}$ on $[0,T_{\max}+\varepsilon']$, because by assumption $H_{q_0}$ is an open set. Then by (c2) and the definition of continuous invariant, we get $\xx(t)\in D_{q_0}'\cap H_{q_0}$ on $[0,T_{\max}+\varepsilon']$, so $T_{\max}$ could not be maximal. Therefore (\ref{eq:r3}) holds. Then there exists $e=(q_0,q')\in E$ such that $\xx(T_{\max})\in G_e'\subseteq D_{q'}'$, so we can make a discrete jump from $q_0$ to $q'$ and extend the hybrid trajectory by continuous evolution at $q'$.

Such extension either ends with a trajectory satisfying $\|\tau\|=\infty$ or goes on forever resulting $\<\tau\rg=\infty$.
\qed
\end{proof}

\section{Proof of Proposition \ref{prop:my-polyhedra}}\label{app:b}
To prove Proposition~\ref{prop:my-polyhedra}, we first introduce the famous \emph{Farkas' lemma} \cite{Tiwari08} in the theory of linear programming.

\begin{lemma}[Farkas' Lemma]\label{lem:farkas}
For linear formulas $p_j,r_k\in \mathbb R[\xx]$, the formula
$\bigwedge_{j\in J} p_j>0 \wedge \bigwedge_{k\in K} r_k\geq 0$ is unsatisfiable (over the reals) if and only if there exist non-negative constants $\mu$, $\mu_j \,(j\in J)$ and $\nu_k \,(k\in K)$ such that
$$\mu+\sum_{j\in J} \mu_j p_j + \sum_{k\in K} \nu_k r_k = 0$$ and
at least one of $\mu_j,\mu$ is strictly positive.
\end{lemma}

\subsection*{Proof of Proposition \ref{prop:my-polyhedra}}
\begin{proof}
Let $\mathbf{q}_i$ denote the $i$-th row vector of the matrix $Q$ and $\rho_i$ denote the $i$-th entry of column vector $\rho$. Then $\mathbf q_i \xx = \rho_i$ stands for the $i$-th facet of the polyhedron $Q\xx \leq \rho$. By a generalization of Theorem \ref{thm:crit} (see \cite{emsoft11} for the detail), $Q\xx \leq \rho$ is a CI of $(\mathbf c\mathbf x < a, A\xx + \mathbf b)$ if for all $1\leq i\leq r$, the following implication holds:

$$Q\xx \leq \rho \wedge \mathbf q_i \xx = \rho_i \wedge \mathbf c\mathbf x < a \Longrightarrow \mathbf q_i (A\xx + \mathbf b) < 0\enspace ,$$
which is equivalent to
\begin{equation}\label{eqn:farkas1}
-Q\xx+\rho\geq 0 \,\wedge \mathbf q_i \xx - \rho_i \geq 0 \,\wedge \mathbf q_i (A\xx + \mathbf b) \geq 0\,\wedge \mathbf {-c}\mathbf x + a >0
\end{equation}
is unsatisfiable.

By Lemma \ref{lem:farkas}, the unsatisfiability of (\ref{eqn:farkas1}) is equivalent to the existence of constant $\gamma_i^i$, and non-negative constants  $\gamma_i^1,\gamma_i^2,\ldots,\gamma_i^{i-1},\gamma_i^{i+1},\ldots\gamma_i^r$, $\eta_i$, $\lambda_i$, $\xi_i$ such that
\begin{equation}\label{eqn:farkas2}
\sum_{1\leq i\leq r}\gamma_i (\mathbf {-q}_i \xx + \rho_i) + \lambda_i \mathbf q_i (A\xx + \mathbf b) + \xi_i (\mathbf {-c}\mathbf x + a) + \eta_i = 0
\end{equation}
and $\xi_i+\eta_i>0$. By equating the coefficients of the left side of (\ref{eqn:farkas2}) to $0$, we get
\begin{itemize}
  \item[(1)] $\sum_{i=1}^{r} (-\gamma_i\mathbf q_i) + \lambda_i\mathbf q_i A - \xi_i \mathbf c = 0$\,;
  \item[(2)] $\sum_{i=1}^{r} \gamma_i\rho_i + \lambda_i\mathbf q_i \mathbf b + \xi_i a + \eta_i = 0$\,; and
  \item[(3)] $\xi_i+\eta_i>0$\,,
\end{itemize}
the matrix form of which corresponds to the three conditions in Proposition \ref{prop:my-polyhedra}.
  \qed
\end{proof}

\end{document}